\documentclass[runningheads,a4paper]{llncs}
\usepackage{bm}
\usepackage{amsfonts}
\usepackage{amsmath}
\usepackage{setspace}
\usepackage{amscd}
\usepackage{amssymb}
\usepackage{graphicx}
\usepackage{mathrsfs}
\usepackage[all]{xy}
\usepackage{verbatim}
\usepackage{cases}
\usepackage{cite}
\begin{document}
\title{ Model Checking : A Co-algebraic Approach}
\titlerunning{Model Checking : A Co-algebraic Approach}  
%
\author{Jianhua Gao$^{1,2}$ \quad Ying Jiang$^1$}
\authorrunning{Jianhua Gao, Ying Jiang}   
%
\tocauthor{Jianhua Gao, Ying Jiang}
\institute{State Key Laboratory of Computer Science\\ Institute of Software, Chinese Academy of Sciences\\Beijing 100190, P.R.China
\and Graduate University of the Chinese Academy of Sciences\\Beijing, 100049, P. R. China\\
\email{\{gaojh, jy\}@ios.ac.cn}}

\maketitle              
\begin{abstract}
State explosion problem is the main obstacle of model checking. In this paper, we try to solve this problem from a coalgebraic approach. We establish an effective method to prove uniformly the existence of the smallest Kripke structure with respect to bisimilarity, which describes all behaviors of the Kripke structures and no redundancy. We show then this smallest Kripke structure generates a concrete smallest one for each given finite Kripke structure and some kind of infinite ones. This method is based on the existence of the final coalgebra of a suitable endofunctor and can be generalized smoothly to other coalgebraic structures. A naive implementation of this method is developed in Ocaml.
\end{abstract}
\section{Introduction}
State explosion problem is the main obstacle of model checking. In this paper, we try to solve this problem from a coalgebraic approach. Actually, we prove the existence of the smallest Kripke structure for the class of Kripke structures over $AP$, where $AP$ is the set of atomic propositions. The smallest Kripke structure describes all behaviors of the class of Kripke structures and there is no redundancy, that is there exists an unique homomorphism $f$ from each Kripke structure $K$ to the smallest Kripke structure such that for any two states $s$, $s'$ of $K$, $f(s) = f(s')$ if and only if $s \sim s'$. We show for each given finite Kripke structure and some infinite ones, this smallest Kripke structure generates a concrete smallest one which is bisimulation equivalent to the original one. Suppose that $M$ is the concrete smallest Kripke structure of $K$, according to \cite{Clarke1999ModelChecking}, for every $CTL^{*}$ formula $f$, $M \models f \Leftrightarrow K \models f$. That means the concrete Kripke structure satisfies the same $(CTL^*)$ properties with the original one. This method is based on the existence of the final coalgebra of a suitable endofunctor. Moreover, for each given Kripke structure, we construct the concrete smallest one by extracting the image of the original one from the smallest Kripke structure. This method can be generalized smoothly to other coalgebraic structures. A naive implementation is developed in Ocaml.

Coalgebra is often viewed as a duality of algebra. Coalgebraic structure is always discussed together with observer. General speaking, infinite and dynamic objects are often coalgebraic. The dual between algebra and coalgebra are represented in \cite{Jacobs1997atutorial}. Coalgebras are general dynamical systems, final coalgebras describe behaviour of such systems (often as infinite objects) in which states and observations coincide, bisimilarity expresses observational indistinguishability \cite{Jacobs}. Different methods of coinduction definition and coinduction proof principle are treated in \cite{Davidebisimulation} by Davide Sangiorgi. A general final coalgebra theorem is given by Peter Aczel and Nax Mendler \cite{Aczel1989Final}. According to this theorem, every set-based functor on the category of classes has a final coalgebra. More concretely, the final coalgebra of the functor of labelled transition system is given by Jan Rutten and Daniele Turi \cite{Rutten1994InitialFinal}. In this paper, the construction of the final coalgebra of the functor of Kripke structures is inspired from \cite{Rutten1994InitialFinal}.

There are many different methods to describe systems with infinite state spaces in finite representation. Among them are context-free processes, Basic Parallel Processes, PA-processes, pushdown processes and Petri nets \cite{Burkart00verificationon,BE96,Mayr1998Infinite,Esp97b}. Our algorithm could not be applied to them directly. Simple graph grammars \cite{QuemenerJ95}, rational Kripke models \cite{Bekker2009Symbolic} are infinite Kripke structures. Our algorithm for infinite Kripke structures is based on simple graph grammars.

The notion of bisimulation plays an important role. For instance, determining the bisimilarity is an efficient method to construct the smallest model. Algorithmic solutions to the bisimilarity on a finite structure are well developed \cite{Clarke1999ModelChecking,Dovier2004Bisimulation}. While minimization algorithm for symbolic bisimilarity are also studied in \cite{Bonchi2009Minimization} for infinite systems. But the latter seems not to be directly applied to infinite Kripke structures.

The paper is organized as follows: in section \ref{section two} we recall some basic notions in category and model checking. Section \ref{section three} we define the functor of Kripke structures. Section \ref{section_final} we construct the final coalgebra for this functor and set the smallest Kripke structure. Section \ref{section five} we construct the concrete smallest Kripke structure. Finally, we give two examples in Section \ref{Examples}.
\section{Preliminaries}\label{section two}
Definitions and theorems of this section about category and model checking are following \cite{Clarke1999ModelChecking, Jacobs, Rutten1994InitialFinal}. The reader is referred to \cite{Pierce1991BasicCategory, Jacobs1997atutorial, Jacobs, Rutten1994InitialFinal, Clarke1999ModelChecking} for further details.
\subsection{Coalgebra}
\begin{definition}[Coalgebra]\label{coalgebra}
Let  $F$ be an endofunctor on category {\em K}. An F-{\em coalgebra} is a pair $(X, \alpha : X \rightarrow F(X))$. An {\em homomorphism} $f: (X, \alpha) \rightarrow (Y, \beta)$ is an arrow $ f :  X  \rightarrow Y $ such that $\beta \circ f = F(f) \circ \alpha$
\end{definition}
\begin{definition}[Bisimulation]\label{bisimulation_coalgebra}
An F-{\em bisimulation} between two F-coalgebras (A, $\alpha$) and (B, $\beta$) is a relation R $\subseteq$ A$\times$B that can be extended to an F-coalgebra (R, $\gamma$), for some $\gamma$ : R $\rightarrow$ F(R), such that its projections $\pi_{1}$ : R$\rightarrow$A and $\pi_{2}$ : R$\rightarrow$ B are homomorphisms of F-coalgebras:
\end{definition}
\[
\begin{CD}
A @<{\pi_{1}}<< R @>{\pi_{2}}>> B\\
@V{\alpha}VV @V{\gamma}VV @VV{\beta}V\\
F(A) @<<{F(\pi_{1})}< F(R) @>>{F(\pi_{2})}> F(B)
\end{CD}
\]
Actually, all of F-coalgebras form a category with coalgebras as objects and bisimulations as arrows.
\begin{definition}[Bisimularity]
 The {\em bisimularity} over an F-coalgebra $(A, \alpha)$, written $\sim_A$, is the union of all bisimulations, that is, $\sim_A=\bigcup\{R \subseteq A \times A| R\ is\ a\ F-bisimulation\ over\ (A, \alpha)\}$
\end{definition}
\begin{definition}[Final Coalgebra]\label{final_coalgebra}
An F-coalgebra (A, $\alpha$) is called {\em final} if for any F-coalgebra (B, $\beta$) there exists an unique homomorphism f:(B, $\beta$)$\rightarrow$ (A, $\alpha$).It is {\em weakly final} if there exists at least one such homomorphism.
\end{definition}
\begin{theorem}\label{final_fixed_points}
Let F be a functor.
\begin{enumerate}
  \item Final coalgebras, if they exist, are uniquely determined (up-to-isomorphism).
  \item Final F-coalgebras (A, $\alpha$) are fixed points of F; that is ,$\alpha$ : A$\rightarrow$ F(A) is an isomorphism.
\end{enumerate}
\end{theorem}
\begin{theorem}\label{extensional}
A final F-coalgebra (A, $\alpha$) is strongly extensional, that is, for all a, $a' \in$ A, if a $\sim_{A}$
$a'$ then a = $a'$.
\end{theorem}
\begin{definition}[Kernel]\label{kernel}
The {\em kernel} of a function h : X $\rightarrow$ Y is the set
\begin{center} $K_{h} =\{(x,x')\in X \times X| h(x) = h(x')\}$.\end{center}
\end{definition}
\begin{definition}[Weakly preserve kernel]
The functor F : \textbf{Set}$\rightarrow$ \textbf{Set} {\em weakly preserves kernels} if $K_{F(f)}$ can be injectively mapped into $F(K_{f})$ for all functions $f$ in \textbf{Set}.
\end{definition}
\begin{theorem}\label{preserves_kernel}
Let F weakly preserve kernels. Let (A, $\alpha$) be a final F-coalgebra and ( B, $\beta$) be any F-coalgebra. Let h be the unique homomorphism from ( B, $\beta$) to ( A, $\alpha$). For all b, $b'$ $\in$ B, b $\sim_{B}$ $b'$ if and only if h(b) = h($b'$).
\end{theorem}
\begin{definition}[Chain]
A {\em chain} in a category \textbf{Set} is a diagram of the following form:\\
$\Delta = X_{0}\xleftarrow{f_{0}}X_{1}\xleftarrow{f_{1}}X_{2}\xleftarrow{f_{2}}\cdots$
\end{definition}
\begin{definition}\label{limit_def}
  A {\em limit} of the chain $\Delta = X_{0}\xleftarrow{f_{0}}X_{1}\xleftarrow{f_{1}}X_{2}\xleftarrow{f_{2}}\cdots$, if it exists, is an object $Z \in Set$ with a collection of arrows $(Z \xrightarrow[]{\zeta_n} X_n)_{n \in \mathbb{N}}$ satisfying $f_{n+1} \circ \zeta_{n+1} = \zeta_n$, with the following universal property. For each object $Y \in Set$ with arrows $g_n: Y \rightarrow X_n$ such that $f_{n+1} \circ g_{n+1} = g_n$, there is an unique map $h: Y \rightarrow Z$ with $\zeta_n \circ h = g_n$, for each $n \in \mathbb{N}$.
\end{definition}
\begin{definition}[$\omega$-continous]\label{continous_def}
Let $Z \in Set$ with a collection of arrows $(Z \xrightarrow[]{\zeta_n} X_n)_{n \in \mathbb{N}}$ be the limit of the chain $\Delta = X_{0}\xleftarrow{f_{0}}X_{1}\xleftarrow{f_{1}}X_{2}\xleftarrow{f_{2}}\cdots$, a functor $F: Set \rightarrow Set$ is said to {\em preserve limit} if $(F(X_n) \xleftarrow[]{F(\zeta_n)} F(Z))_{n \in \mathbb{N}}$ is the limit of the chain $\Delta' = F(X_{0})\xleftarrow{F(f_{0})}F(X_{1})\xleftarrow{F(f_{1})}F(X_{2})\xleftarrow{F(f_{2})}\cdots$ resulting from applying $F$. The functor $F$ called  $\omega$-{\em continuous} if it preserves limits of all chains. Another way to formulate this is: the induced map $F(Z) \rightarrow Z$ is an isomorphism.
\end{definition}
\begin{theorem}\label{limit_exists_set}
 For a chain $\Delta = D_{0}\xleftarrow{f_{0}}D_{1}\xleftarrow{f_{1}}D_{2}\xleftarrow{f_{2}}\cdots$ in \textbf{Set}, the limit $Z$ of the chain is a
subset of the infinite product $\prod\limits_{n\in \mathbb{N}}{D_{n}}$ given by, Z = $\{(d_{0},d_{1},d_{2},\ldots)|$ $\forall n \in \mathbb{N}\ d_{n} \in D_{n} \wedge f_{n}(d_{n+1})=d_{n}\}$.
\end{theorem}
\begin{theorem}\label{final_exists_continuous}
Each $\omega$-continuous endofunctor $F : \textbf{Set} \rightarrow \textbf{Set}$ has a final coalgebra, obtained as limit $Z \rightarrow F(Z)$ of the chain $(F^{n+1}(1) \xrightarrow{F^{n}(!)} F^{n}(1))_{n \in \mathbb{N}}$
\end{theorem}
\begin{theorem}\label{final_exists_pf(a)pf(x)}
Let F and G be two functors from \textbf{Set} to \textbf{Set}. Let $\pi$: F$ \xrightarrow{.}$G be a natural transformation: a family of functions $\{\pi_{X}\}$--one for each set X--with, for any function $f : X \rightarrow Y$, we have $\pi_Y \circ F(f) = G(f) \circ \pi_X$.
Suppose $\pi_{X}$ is surjective, for any set X. If F has a final F-coalgebra, then also G has a final G-coalgebra.
\end{theorem}
\subsection{Model Checking}
\begin{definition}[Kripke structure]\label{kripke_structure}
Let AP be a set of atomic proposition. A {\em Kripke structure M} over AP is a quadruple $M =
(S,S_{0},R,L)$ where
\begin{enumerate}
\item S is a finite set of states;
\item $S_{0}\subseteq S$ is the set of initial states;
\item $R \subseteq S \times S$ is a transition relation that must be total, that is, for every state $s \in S$ there is a state $s' \in S$ such that $R(s,s')$;
\item $L : S \rightarrow 2^{AP}$ is a function that lables each state with the set of atomic propositions ture in that state.
\end{enumerate}
\end{definition}
\begin{definition}[Bisimulation relation]\label{bisimulation_relation}
Let $M=( S, R, S_{0}, L)$ and $M'=(S', R', S_{0}^{'}, L')$ be two Kripke structures with the same set of atomic propositions AP. A relation $B\subseteq S\times S'$ is a {\em bisimulation relation} between M and $M'$ if and only if for all $s \in S$ , $s' \in S'$, if $B(s, s')$ then the following conditions hold:
\begin{enumerate}
\item  $L(s)=L'(s')$;
\item  For every state $s_{1} \in S$ such that $R(s, s_{1})$ there is $s'_{1} \in S'$ such that $R'(s', s'_{1})$ and $B(s_{1}, s'_{1})$;
\item  For every state $s'_{1} \in S'$ such that $R'(s', s'_{1})$ there is $s_{1} \in S$ such that $R(s, s_{1})$ and $B(s_{1}, s'_{1})$.
\end{enumerate}
\end{definition}
\begin{definition}[Bisimulation equivalent]\label{bisimulation_equivalent}
The structures M and $M'$ are {\em bisimulation equivalent}(denoted by $M\equiv M'$) if there exists a bisimulation relation B such that for every initial state $s_{0}\in S_{0}$ in M there is an initial state $s'_{0}\in S'_{0}$ in $M'$ such that $B(s_{0}, s'_{0})$. In addition, for every initial state $s'_{0}\in S'_{0}$ in $M'$ there is an initial state $s_{0}\in S_{0}$ in M such that $B(s_{0}, s'_{0})$.
\end{definition}
\begin{theorem}\label{structure_meet}
If $M \equiv M'$ iff for every $CTL^{*}$ formula $\psi$, $M \models \psi \Leftrightarrow M' \models \psi$.
\end{theorem}
\section{The Functor of Kripke Structure}\label{section three}
\subsection{Functor}\label{functor}
 Nothing in the essence of the following approach requires the Kripke structure to be finite. Actually, Kripke structures for real systems are very often infinite. The finiteness constraint is due to our current technology, not to the approach itself. Here, we do allow the set of states of Kripke structures  be infinite. In this section we introduce a particular functor $\mathscr{P}(AP) \times \mathscr{P}(\cdot) : \textbf{Set} \rightarrow \textbf{Set}$, which can express all Kripke structures (the set of states may be infinite)over $AP$. It is defined as follows.
\begin{center}$\mathscr{P}(AP) \times \mathscr{P}(S)=\{(A,V)|  A \subseteq AP, V \subseteq S, V \neq \phi\}$\end{center}
\begin{center}
  ($V \neq \phi$ is because that every state should have at least one successor state)
\end{center}
$\mathscr{P}(AP) \times \mathscr{P}(\cdot)$ maps a function $f : S \rightarrow T$ to the function $\mathscr{P}(AP) \times \mathscr{P}(f):\mathscr{P}(AP) \times \mathscr{P}(S) \rightarrow \mathscr{P}(AP) \times \mathscr{P}(T)$, which is defined as follows.
\begin{center}$\mathscr{P}(AP) \times \mathscr{P}(f)(A,  V)=(A,\{ f(s)\in T|s \in V\})$\end{center}
\begin{definition}[coalgebraic Kripke structure]
Let $\mathcal{A}$ be a $\mathscr{P}(AP)\times \mathscr{P}(\cdot)$-coalgebra $(A, \alpha)$ and $I \subseteq A$, a {\em coalgebraic Kripke structure} over $AP$ is a pair $(\mathcal{A} , I)$.
\end{definition}
\begin{proposition}\label{one-to-one-correspondence}
Let $\mathcal {K}$ be the class of Kripke structures over $AP$ and $\mathcal {C}$ be the class of coalgebraic Kripke structures over $AP$, then $\mathcal {K} \cong \mathcal {C}$.
\end{proposition}
\begin{proof}
It is sufficient to show that there exist two functions $f: \mathcal {K} \rightarrow \mathcal {C}$ and $g: \mathcal {C} \rightarrow \mathcal {K}$ such that $f \circ g = id_{\mathcal {C}}$ and $g \circ f = id_{\mathcal {K}}$.
\begin{itemize}
  \item We firstly construct the function $f$, mapping each Kripke structure $M = (S,S_{0},R,L)$ to coalgebraic Kripke structure $(\mathcal{A} , I)$ defined as follows:
    \begin{enumerate}
      \item $I = S_{0}$.
      \item $\mathcal{A} = (S,\gamma)$, where $\gamma(s) = (L(s),\{s' | (s,s') \in R\})$
    \end{enumerate}
  \item We then construct the function $g$, mapping each coalgebraic Kripke structure $((A, \gamma), I)$ to Kripke structure $M = (S,S_{0},R,L)$ defined as follows:
    \begin{enumerate}
      \item $S = A$.
      \item $S_{0} = I$.
      \item $R=\{(s_{1},s_{2}) | s_{2} \in \pi_{2}(\gamma(s_{1}))\}$.
      \item $L(s) = \pi_{1}(\gamma(s))$.
    \end{enumerate}
  \item Finally, we need to show that $f \circ g = id_{\mathcal {C}}$ and $g \circ f = id_{\mathcal {K}}$. According to the constructions of $f$ and $g$, it is clearly true.
\end{itemize}
\end{proof}
\begin{example}
  Let $M$ be the Kripke structure as the following figure, we construct a coalgebraic Kripke structure $((S,\gamma), I)$ defined by $S = \{s_1, s_2\}$, $\gamma (s_1) = (\{a\},\{s_{2}\})$, $\gamma (s_2) = (\{b\},\{s_{1}\})$ and $I = \{s_1\}$.
  \begin{center} \includegraphics[width=0.20\textwidth]{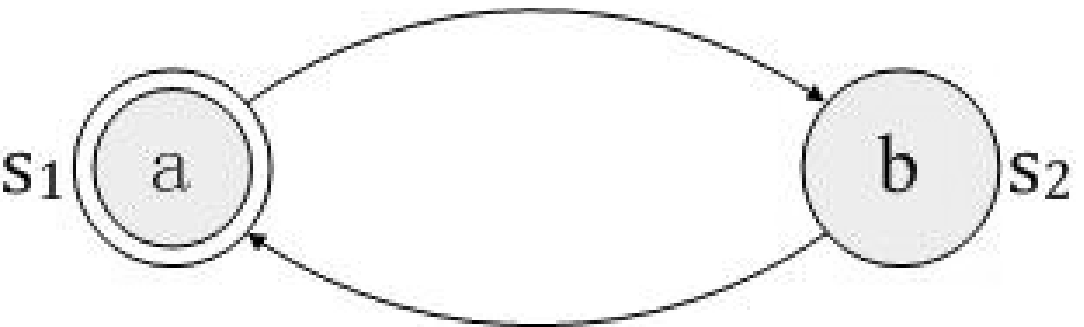}\end{center}
\end{example}

Notice that there does not exist a final coalgebra for the functor $\mathscr{P}(AP)\times \mathscr{P}(\cdot) :\textbf{Set} \rightarrow \textbf{Set}$. According to Theorem \ref{final_fixed_points}, any final coalgebra is a fixed point. However, it is well known the functor $\mathscr{P}(AP)\times \mathscr{P}(\cdot)$ does not have any fixed points.
Therefore we need to restrict the functor $\mathscr{P}(AP)\times \mathscr{P}(\cdot)$ to the functor $\mathscr{P}_{f}(AP)\times \mathscr{P}_{f}(\cdot) : \textbf{Set} \rightarrow \textbf{Set}$ defined as follows:
\begin{center}$\mathscr{P}_{f}(AP)\times \mathscr{P}_{f}(S)=\{(V,W)\in \mathscr{P}(AP)\times \mathscr{P}(S)|V,W$ are finite$\}$;\end{center}
For simplicity, we denote this functor by $\mathcal {P}_f$ from now on. This functor can represents Kripke Structures with finite branch, that means every state should have finite successor states, however the set of states may be infinite.
\begin{proposition}\label{p_f_weakly_preserve_kernel}The functor $\mathcal{P}_f$ weakly preserves kernels.\end{proposition}
\begin{proof}
Given $f : X \rightarrow Y$, then $\mathcal{P}_f(f) : \mathcal{P}_f(X) \rightarrow \mathcal{P}_f(Y)$. $K_{f} = \{(x_{1},x_{2}) | f(x_{1}) = f(x_{2})\}$, $K_{\mathcal{P}_f(f)} = \{((S,V_{1}),(S,V_{2})) | S \subseteq AP, \forall v_{1} \in V_{1} \exists v_{2} \in V_{2} s.t. f(v_{2})=f(v_{1}),\forall v_{2} \in V_{2} \exists v_{1} \in V_{1} s.t. f(v_{1})=f(v_{2})\}$. We can construct an injection $i : K_{\mathcal{P}_f(f)} \rightarrow \mathcal{P}_f(K_{f})$ as follows:
$i : ((S,V_{1}),(S,V_{2})) \longmapsto (S,\{(v_{1},v_{2}) | v_{1} \in V_{1}, v_{2} \in V_{2}, f(v_{1}) = f(v_{2})\})$. Notably, $i$ is an injection, so $\mathcal{P}_f$ weakly preserves kernels.
\end{proof}
\subsection{Bisimulations}
In this subsection we will show that bisimulation relations of Kripke structures over $AP$ coincide with bisimulations of coalgebras of the functor $\mathcal {P}_f$.
\begin{definition}\label{coalgebraic_bisimulation_equ}
  Let $(\mathcal{A} , I)$, $(\mathcal{A}' , I')$ be two coalgebraic Kripke structure, we said $(\mathcal{A} , I)$ bisimulation equivalent to $(\mathcal{A}' , I')$ (denoted by $(\mathcal{A} , I) \doteq (\mathcal{A}' , I')$ )if there exists an bisimulation $B$ from $\mathcal{A}$ to $\mathcal{A}'$ such that for every $s \in I$ there exists an $s' \in I'$ and $(s, s') \in B$, and for every $s' \in I'$ there exists an $s \in I$ and $(s, s') \in B$.
\end{definition}
\begin{proposition}\label{kripke_structure_coalgebra}
Let $K$, $K'$ be two Kripke Structures and $(\mathcal {A}_{K}, I_K)$, $(\mathcal {A}_{K'}, I_{K'})$ be the corresponding coalgebraic Kripke structures, then $K \equiv K'$ iff $(\mathcal{A}_{K} , I_K) \doteq (\mathcal{A}_{K'}, I_{K'})$.
\end{proposition}
\begin{proof}
Suppose that $K = (S, R, S_0, L)$, $K' = (S', R', S'_0, L')$. Recall the proof of Proposition \ref{one-to-one-correspondence}, there exist $\alpha$, $\alpha'$ such that $\mathcal {A}_{K} = (S, \alpha)$ and $\mathcal {A}_{K'} = (S', \alpha')$.

It is sufficient to show that for every $B \in S \times S'$, $B$ is a bisimulation relation between $K$, $K'$ iff $B$ is a bisimulation between $\mathcal {A}_{K}$ and $\mathcal {A}_{K'}$.

We firstly prove that if $B$ is a bisimulation relation between $K$ and $K'$ then it is a bisimulation between $\mathcal {A}_{K}$ and $\mathcal {A}_{K'}$. We extend $B$ to $(B, \gamma)$, where $\gamma(s_{1},s_{2}) = (L(s_{1}), \{(s'_{1},s'_{2}) \in B | (s_{1}, s'_{1}) \in R, (s_{2},s'_{2}) \in R'\})$. Now we need to show that $(B, \gamma)$ can make the following diagram commutes.
\begin{displaymath}
\xymatrix{
S \ar[d]_{\alpha} &B \ar[l]_{\pi_{1}} \ar[r]^{\pi_{2}} \ar[d]_{\gamma} &S' \ar[d]^{\alpha'}\\
\mathcal{P}_f(S) &\mathcal{P}_f(B)\ar[l]^{\mathcal{P}_f(\pi_{1})} \ar[r]_{\mathcal{P}_f(\pi_{2})} &\mathcal{P}_f(S') }
\end{displaymath}
Left square commutes, since for any element $(s_{1}, s_{2}) \in B$
\[ \begin{split}
\alpha \circ \pi_{1}(s_{1}, s_{2}) &= \alpha(s_{1})\quad (definition\ of\ \pi_{1})\\
&= (L(s_{1}), \{s'_{1} | (s_{1}, s'_{1}) \in R\})\quad (proof\ of\ Proposition\ \ref{one-to-one-correspondence})
\end{split} \]
\[ \begin{split}
\mathcal{P}_f(\pi_{1}) \circ \gamma(s_{1}, s_{2}) &= \mathcal{P}_f(\pi_{1})(L_{1}(s_{1}),\{(s'_{1},s'_{2}) \in B | (s_{1}, s'_{1}) \in R, (s_{2},s'_{2}) \in R'\})\quad \\
&= (L_{1}(s_{1}), \{s'_{1} | \exists s'_{2},(s'_{1},s'_{2}) \in B,(s_{1}, s'_{1}) \in R, (s_{2},s'_{2}) \in R'\})\quad \\
&= (L_{1}(s_{1}), \{s'_{1} | (s_{1}, s'_{1}) \in R\})\quad (definition\ of\ B)
\end{split} \]
So $\alpha \circ \pi_{1} = \mathcal{P}_f(\pi_{1}) \circ \gamma$, that is the left square commutes. Without loss of generality, the right square also commutes.

We then prove that if $B$ is a bisimulation between $\mathcal {A}_{K}$ and $\mathcal {A}_{K'}$ then it is a bisimulation relation between $K$ and $K'$.
We need to prove that for any element $(s_{1}, s_{2}) \in B$,
\begin{enumerate}
  \item $L(s_{1}) = L'(s_{2})$;
\[ \begin{split}
L(s_{1}) &= \pi_1(\alpha(s_{1}))\quad (the\ proof\ of\ Proposition\ \ref{one-to-one-correspondence})\\
&= \pi_1(\alpha \circ \pi_{1}(s_{1}, s_{2}))\quad (definition\ \pi_{1})\\
&= \pi_1(\mathcal{P}_f(\pi_{1}) \circ \gamma(s_{1}, s_{2}))\quad (left\ saquare\ commutes)\\
&= \pi_1(\mathcal{P}_f(\pi_{2}) \circ \gamma(s_{1}, s_{2}))\quad (definition\ \mathcal{P}_f)\\
&= \pi_1(\beta \circ \pi_{2}(s_{1}, s_{2}))\quad (right\ saquare\ commutes)\\
&= \pi_1(\beta(s_{2}))\quad (definition\ \pi_{2})\\
&= L'(s_{2})\quad (the\ proof\ of\ Proposition\ \ref{one-to-one-correspondence})
\end{split} \]
  \item  For every state $s'_{1}$ such that $R(s_{1}, s'_{1})$ there is $s'_{2}$ such that $R'(s_{2}, s'_{2})$ and $B(s'_{1}, s'_{2})$; For every state $s'_{2}$ such that $R'(s_{2}, s'_{2})$ there is $s'_{1}$ such that $R(s_{1}, s'_{1})$ and $B(s'_{1}, s'_{2})$.
\[ \begin{split}
R(s_{1}) &= \pi_2(\alpha(s_{1}))\quad (the\ proof\ of\ Proposition\ \ref{one-to-one-correspondence})\\
&= \pi_2(\alpha \circ \pi_{1}(s_{1}, s_{2}))\quad (definition\ \pi_{1})\\
&= \pi_2(\mathcal{P}_f(\pi_{1}) \circ \gamma(s_{1}, s_{2}))\quad (the\ left\ square\ commutes)\\
&= \{s'_{1} | \exists s'_{2}\ s.t.\ (s'_{1}, s'_{2}) \in \pi_2(\gamma(s_{1}, s_{2}))\} \quad (definition\ \mathcal{P}_f)\\
&= \{s'_{1} | \exists s'_{2}\ s.t.\ s'_{2} \in \pi_2(\mathcal{P}_f(\pi_{2}) \circ \gamma(s_{1}, s_{2}))\} \quad (definition\ \mathcal{P}_f)\\
&= \{s'_{1} | \exists s'_{2}\ s.t.\ s'_{2} \in \pi_2(\beta \circ \pi_{2}(s_{1}, s_{2}))\ and\  (s'_{1}, s'_{2}) \in B\} \quad (the\ right\ square\ commutes)\\
&= \{s'_{1} | \exists s'_{2}\ s.t.\ s'_{2} \in \pi_2(\beta(s_{2}))\ and\  (s'_{1}, s'_{2}) \in B\} \quad (definition\ \pi_{2})\\
&= \{s'_{1} | \exists s'_{2}\ s.t.\ s'_{2} \in R'(s_{2})\ and\  (s'_{1}, s'_{2}) \in B\} \quad (the\ proof\ of\ Proposition\ \ref{one-to-one-correspondence})\\
Without&\ loss\ of\ generality,\ R'(s_{2}) = \{s'_{2} | \exists s'_{1}\ s.t.\ s'_{1} \in R(s_{1})\ and\  (s'_{1}, s'_{2}) \in B\}.
\end{split}\]
\end{enumerate}
\end{proof}
\section{The Smallest Kripke Structure}\label{section_final}
In this section we mainly construct the final coalgebra of the functor $\mathcal{P}_f$. We firstly prove the existence of the final coalgebra of this functor by using Theorem \ref{final_exists_pf(a)pf(x)}.
\subsection{Existence of the Final Coalgebra of the Functor $\mathcal{P}_f$}
To use Theorem \ref{final_exists_pf(a)pf(x)}, we need to find another functor $F : \bm{Set} \rightarrow \bm{Set}$ of which final coalgebra exists and there exists a nature transformation $\pi : F \xrightarrow{.} \mathcal{P}_f$ such that $\pi_X$ is surjective for every set $X$.

Set the functor $F : \bm{Set} \rightarrow \bm{Set}$, for a set $X$, by $$F(X)=\mathop\sum \limits_{1\leq n< \omega}X^{n}+(AP\times X)^{n}$$ and for an arrow $f: X \rightarrow Y$, by $F(f) : F(X) \rightarrow F(Y)$
\makeatletter
\let\@@@alph\@alph
\def\@alph#1{\ifcase#1\or \or $'$\or $''$\fi}\makeatother
\begin{subnumcases}
{F(f)(Z)=}
(f(x_1), f(x_2), \ldots, f(x_n)), &if $Z = (x_1, x_2, \ldots, x_n)$,\nonumber\\
((a_1,f(x_1)), \ldots, (a_m,f(x_m))), &if $Z = ((a_1,x_1), \ldots, (a_m,x_m))$.\nonumber
\end{subnumcases}
\makeatletter\let\@alph\@@@alph\makeatother

Now, we prove that $F$ has a final coalgebra. Indeed the functor $F$ is $\omega$-continuous. For every $\omega$-chain in \textbf{Set}:
  \begin{center}
    $\Delta = D_{0} \xleftarrow{f_{0}} D_{1} \xleftarrow{f_{1}} D_{2} \xleftarrow{f_{2}} \cdots$
  \end{center}
  The limit of $\Delta$ is $Z = \{(d_{0},d_{1},d_{2},\ldots)|\forall n \in \mathbb{N}.d_{n} \in D_{n} \wedge f_{n}(d_{n+1})=d_{n}\}$ to use Theorem \ref{limit_exists_set}.
  Applying $F$ to $\Delta$ we can obtain the $\omega$-chain $F(\Delta)$:
  \begin{center}
    $F(\Delta) = F(D_{0}) \xleftarrow{F(f_{0})} F(D_{1}) \xleftarrow{F(f_{1})} F(D_{2}) \xleftarrow{F(f_{2})} \cdots$
  \end{center}
  Similarly, the limit of $F(\Delta)$ is $Z' = \{(d'_{0},d'_{1},d'_{2},\ldots)|\forall n \in \mathbb{N}.d'_{n} \in F(D_{n}) \wedge F(f_{n})(d'_{n+1})=d'_{n}\}$.
  It is sufficient to show that $F(Z) \cong Z'$. We can construct a function $\gamma : F(Z) \rightarrow Z'$ as follows:
  \begin{enumerate}
    \item $\gamma(x) = ((d^{0}_{0},d^{1}_{0},\ldots,d^{n}_{0}),(d^{0}_{1},d^{1}_{1},\ldots,d^{n}_{1}),\ldots)$\\ $(x=((d^{0}_{0},d^{0}_{1},\ldots),\ldots,(d^{n}_{0},d^{n}_{1},\ldots)) \in Z^n)$
    \item $\gamma(x) = (((a_{0},d^{0}_{0}),(a_{1},d^{1}_{0}),\ldots,(a_{n},d^{n}_{0})),((a_{0,}d^{0}_{1}),(a_{1},d^{1}_{1}),\ldots,(a_{n},d^{n}_{1})),\ldots)$\\ $(x=((a_{0},(d^{0}_{0},d^{0}_{1},\ldots)),\ldots,(a_{n},(d^{n}_{0},d^{n}_{1},\ldots))) \in (AP \times Z)^n)$
  \end{enumerate}
  We need to construct another function $\theta : Z' \rightarrow F(Z)$ as follows:
    \begin{enumerate}
    \item $\theta(x) = ((d^{0}_{0},d^{0}_{1},\ldots),\ldots,(d^{n}_{0},d^{n}_{1},\ldots))$\\
    $(x=((d^{0}_{0},d^{1}_{0},\ldots,d^{n}_{0}),(d^{0}_{1},d^{1}_{1},\ldots,d^{n}_{1}),\ldots),\ where\ (d^{0}_{i},d^{1}_{i},\ldots,d^{n}_{i}) \in (D_{i})^n)$
    \item $\theta(x) = ((a_{0},(d^{0}_{0},d^{0}_{1},\ldots)),\ldots,(a_{n},(d^{n}_{0},d^{n}_{1},\ldots)))$\\
    $(x= (((a_{0},d^{0}_{0}),(a_{1},d^{1}_{0}),\ldots,(a_{n},d^{n}_{0})),((a_{0,}d^{0}_{1}),(a_{1},d^{1}_{1}),\ldots,(a_{n},d^{n}_{1})),\ldots),\\ where ((a_{0},d^{0}_{i}),(a_{1},d^{1}_{i}),\ldots,(a_{n},d^{n}_{i})) \in (AP \times D_{i})^n)$
  \end{enumerate}
  It is notable that $\theta \circ \gamma = id_{F(Z)}$, $\gamma \circ \theta = id_{Z'}$, so $F(Z) \cong Z'$, that is F is $\omega$-continuous. To use Theorem \ref{limit_exists_set}, $F$ has a final coalgebra. To use Theorem \ref{final_exists_continuous} the final coalgebra $(T, \gamma)$ is as follows: its state space $T$ is the limit of the chain:
\begin{center}
  $1 \xleftarrow{!} F(1) \xleftarrow{F(!)} F^{2}(1) \xleftarrow{F^{2}(!)} \ldots$
\end{center}
It consists of all finitely branching, ordered trees over {\small$\bigcup\limits_{0\leq n< \omega}AP^{n}$}, of which each branch is an infinite path. Its structure $\gamma : T \rightarrow F(T)$ sends a tree $t \in T$ with $(a_1, a_2, \ldots, a_n)$ as its root and $(t_1, t_2, \ldots, t_n)$ as its immediate subtrees to $((a_1, t_1), (a_2, t_2), \ldots, (a_n, t_n))$.

Then, we define $\pi: F \xrightarrow{.} \mathcal{P}_f$. We difine the family of functions $\pi_X : F(X)\rightarrow \mathcal{P}_f(X)$, for any set $X$, defined by :
\makeatletter
\let\@@@alph\@alph
\def\@alph#1{\ifcase#1\or \or $'$\or $''$\fi}\makeatother
\begin{subnumcases}
{\pi_X(x)=}
(\emptyset, \{x_{0},x_{1},\ldots,x_{n}\}), &if $x=(x_{0},x_{1},\ldots,x_{n}) \in X^n$,\nonumber\\
(\{a_{0},\ldots,a_{n}\},\{x_0, \ldots, x_n\}) &if $x=((a_{0},x_{0}), \ldots,(a_{n},x_{n})) \in (AP \times X)^n$.\nonumber
\end{subnumcases}
\makeatletter\let\@alph\@@@alph\makeatother
Notably, this defines a natural transformation $\pi : F  \xrightarrow{.} \mathcal{P}_f$, such that each $\pi_X$ is surjective. Finally, Theorem \ref{final_exists_pf(a)pf(x)} ensures the existence of a final coalgebra for functor $\mathcal{P}_f$.
\subsection{Concrete Final Coalgebra of the Functor $\mathcal{P}_f$}
Before constructing the concrete final coalgebra of the functor $\mathcal{P}_f$, we introduce some notations following \cite{Rutten1994InitialFinal}. Let $(S, \alpha)$ be a $\mathcal{P}_f$-coalgebra, $R$ be a $\mathcal{P}_f$-bisimulation, we denote for any $s \in S$, $[s]_{R}=\{s'\in S\ |\ (s,s')\in R\}$; $S_{R}=\{[s]_{R}\ |\ s \in S\}$; $\xi_{R}: S \rightarrow S_{R}$ is defined by $\xi_{R}(s)= [s]_{R}$; $\alpha_{R} : S_{R} \rightarrow  \mathcal{P}_f(S_{R})$ is defined by $\alpha_{R}([s]_{R})=(\pi_1(\alpha(s)), \{[s']_R\ |\ s' \in \pi_2(\alpha(s))\})$. In particular, $\xi_{\sim}$ is an homomorphism; $(S_{\sim},\alpha_{\sim})$ is strongly extensional; For any $\mathcal{P}_f$-coalgebra $(D,d)$, if $f , g : (D,d) \rightarrow (S_{\sim}, \alpha_{\sim})$ are two homomorphisms, then $f=g$.

Now, we prove actually the final coalgebra of $\mathcal{P}_f$ is $(T_{\thicksim}, (\pi_{T} \circ \gamma)_{\thicksim})$ written as $(P, \psi)$ for simplicity, where $(T, \gamma)$ is the final coalgebra of the functor $F$. It is sufficient to prove that :\begin{enumerate}\item $(P, \psi)$ is a $\mathcal{P}_f$-coalgebra.
  \item for any $\mathcal{P}_f$- coalgebra $(X, \alpha)$, there is a $\mathcal{P}_f$-homomorphism $h : X \rightarrow P$
  \item for any $\mathcal{P}_f$- coalgebra $(X, \alpha)$, if $f , g : X \rightarrow P$ be two $\mathcal{P}_f$-homomorphisms, then $f = g$.\end{enumerate}$$\xymatrix{X \ar[d]_{\rho \circ \alpha} \ar[r]^{f} & T \ar[d]_{D_1\ \ \ \ \ }^{\gamma} \ar[r]^{\xi_{\sim}} & T_{\sim} \ar[dd]_{D_3\ \ \ \ \ \ }^{(\pi_T \circ \gamma)_{\sim}} \\
  F(X)\ar[d]_{\pi_X} \ar[r]^{F(f)} & F(T) \ar[d]_{D_2\ \ \ \ \ }^{\pi_T} \\
  \mathcal{P}_f(X) \ar[r]_{\mathcal{P}_f(f)} & \mathcal{P}_f(T) \ar[r]_{\mathcal{P}_f(\xi_{\sim})} & \mathcal{P}_f(T_{\sim})}$$
 Here, we show the second point which is the only no trivial case. Let $(X, \alpha)$ be any $\mathcal{P}_f$- coalgebra, because $\pi_{X}$ is surjective, there exists a function $\rho :  \mathcal{P}_f(X) \rightarrow F(X)$ such that $\pi_{X} \circ \rho = Id_{\mathcal{P}_f(X)}$ and $(X, \rho \circ \alpha)$ is an $F$- coalgebra. According to the definition of final coalgebra, there exists an $F$-homorphism $f$ from $(X, \rho \circ \alpha)$ to $(T, \gamma)$ which is the final of the functor $F$, that means the square $D_1$ commutes. To use the definition of nature transformation, the square $D_2$ also commutes. So does the left big square. Also because $\xi_{\sim} : T \rightarrow T_{\sim}$ is a $\mathcal{P}_f$-homomorphism between $(T, \pi_{T} \circ \gamma)$ and $ (T_{\thicksim}(\pi_{T} \circ \gamma)_{\thicksim})$, that is the diagram $D_3$ commutes. Therefore, the whole square commutes. That means $\xi_{\sim} \circ f : X \rightarrow T_{\sim}$ is a $\mathcal{P}_f$-homomorphism from $(X, \alpha)$ to $(T_{\thicksim}, (\pi_{T} \circ \gamma)_{\thicksim})$.
\begin{definition}[The smallest Kripke structure]\label{final_Kripke_structure}
Let $(P, \psi)$ be the final coalgebra of functor $\mathcal{P}_f$, the Kripke structure corresponding to the coalgebraic Kripke structure $((P, \psi), P)$  is called the {\em the smallest Kripke structure} of the class Kripke structures over $AP$.
\end{definition}\begin{corollary}There is only one {\em smallest Kripke structure} over $AP$.\end{corollary}\begin{proof}From Theorem \ref{final_fixed_points} and the proof of Proposition \ref{one-to-one-correspondence}.\end{proof}
\begin{definition}\label{homomorphism_Kripke}
  Let $K = (S, S_0, R, L)$, $K' = (S', S'_0, R', L')$ be two Kripke structures over $AP$, we called $h: S \rightarrow S'$ is an homomorphism from $K$ to $K'$, if for each $s \in S$ the following conditions hold: $s \in S_0 \Rightarrow h(s) \in S'_0$, $L(s) = L'(h(s))$, $h(R(s)) = R'(h(s))$.
\end{definition}
\begin{proposition}\label{Unique_homomorphism_final_kripke}
  Let $K_f = (S_f, S_{0_{f}}, R_f, L_f)$ be the smallest Kripke Structure over $AP$, for any Kripke structure $K = (S, S_0, R, L)$ over $AP$ there exists an unique homomorphism $f : S \rightarrow S_f$ from $K$ to $K_f$.
\end{proposition}
\begin{proof}
  We firstly show that there exists an homomorphism $f: S \rightarrow S_f$ from $K$ to $K_f$ (Existence). Then, show it is unique (Uniqueness).

  According to the proof of Proposition \ref{one-to-one-correspondence}, for Kripke structures $K$, $K_f$, there exist coalgebraic Kripke structures $((S, \alpha), S_0)$, $((S_f, \alpha_f), S_{0_f})$ respectively, where $\alpha : S \rightarrow \mathcal{P}_f(S)$ is defined by $s \mapsto (L(s),\{s' | (s,s') \in R\})$ and $\alpha_f : S_f \rightarrow \mathcal{P}_f(S_f)$ is defined by $s \mapsto (L_f(s),\{s'\ |\ (s,s') \in R_f\}$. To use Definition \ref{final_Kripke_structure}, $(S_f, \alpha_f)$ is the final coalgebra of the functor $\mathcal{P}_f$. According to Definition \ref{final_coalgebra}, there exists an unique homomorphism $h : S \rightarrow S_f$ from $(S, \alpha)$ to $(S_f, \alpha_f)$.
  \begin{itemize}
    \item Existence. We now try to prove that $h : S \rightarrow S_f$ is also an homomorphism from $K$ to $K_f$. It is sufficient to prove that for all $s \in S$, $s \in S_0 \Rightarrow h(s) \in S_{0_f}$; $L(s) = L(h(s))$; $h(R(s)) = R_f(h(s))$.
        \begin{enumerate}
          \item $s \in S_0 \Rightarrow h(s) \in S_f$, $S_f = S_{0_f}$ so $s \in S_0 \Rightarrow h(s) \in S_{0_f}$;
          \item $L(s) = \pi_1(\alpha (s)) = \pi_1(\mathcal{P}_f(h) \circ \alpha (s)) = \pi_1(\alpha_f \circ h(s)) = L(h(s))$;
          \item $h(R(s)) = h(\{s'\ |\ (s, s') \in R\}) = h(\{s'\ |\ s' \in \pi_{2}(\alpha(s))\}) = \pi_2(\mathcal{P}_f(h) \circ \alpha(s)) = \pi_2(\alpha_f \circ h(s)) = R_f(h(s))$.
        \end{enumerate}
    \item Uniqueness. It is sufficient to prove that if $h: S \rightarrow S_f$ is the homomorphism from $K$ to $K_f$, then it is also an homomorphism from $(S, \alpha)$ to $(S_f, \alpha_f)$. For every $s \in S$, $\mathcal{P}_f(h) \circ \alpha(s) = \mathcal{P}_f(h)((L(s), R(s))) = (L(s), h(R(s))) = (L(s), R_f(h(s))) = \alpha_f \circ h(s)$. So $h: S \rightarrow S_f$ is an homomorphism from $(S, \alpha)$ to $(S_f, \alpha_f)$.
  \end{itemize}
\end{proof}
\begin{theorem}
  Let $K = (S, S_0, R, L)$ be any Kripke structure over $AP$, $K_f = (S_f, S_{0_f}, R_f, L_f)$ be the smallest Kripke structure over $AP$, $f : S \rightarrow S_f$ be the unique homomorphism from $K$ to $K_f$. For all $s$, $s' \in S$ if $f(s) = f(s')$, then $s \sim_K s'$, where $\sim_k$ is the bisimularity over $K$.
\end{theorem}
\begin{proof}
  From Theorem \ref{preserves_kernel}, Proposition \ref{p_f_weakly_preserve_kernel}, the proof of Proposition \ref{kripke_structure_coalgebra}, the proof of Proposition \ref{Unique_homomorphism_final_kripke}.
\end{proof}
\section{The Concrete Smallest Kripke Structure}\label{section five}
In this section we will prove that for every Kripke structure $K$ there exists a concrete smallest Kripke structure $M$. This concrete smallest Kripke structure describes all behavior of $K$ and no redundancy, moreover $M \equiv K$.
\begin{proposition}\label{bisimulation_equivalent_ralation}
  $\equiv$, $\doteq$ are equivalent relations.
\end{proposition}
\begin{proof}
  The proof is trivial. Here, we omit it .
\end{proof}
\begin{lemma}\label{bisimulation_equal_kernel}
   Let $(B, \beta)$ be any $\mathcal{P}_f$-coalgebra, $(A, \alpha)$ be the final coalgebra of $\mathcal{P}_f$. If $h : B \rightarrow A$ is the unique homomorphism from $(B, \beta)$ to $(A, \alpha)$, then $\sim_{B} = K_{h}$
\end{lemma}
\begin{proof}
$(b, b') \in \sim_{B}$ iff $b \sim_{B} b'$ iff $h(b) = h(b')$ iff $(b, b') \in K_{h}$.
\end{proof}
\begin{definition}[The concrete smallest Coalgebraic Kripke structures]
   The coalgebraic Kripke structure $((C, \alpha), I)$ is
  \begin{enumerate}
    \item {\em reduced} if $\thicksim_{C} = \{(a,a) | a \in C\}$ (extensional)
    \item {\em connected} if every state s of C is reachable, i.e., there exists a path $s_{0},s_{1},\ldots,s_{n}$ s.t. $s_{0} \in I$, $s_{i+1} \in \pi_2(\alpha(s_{i}))$, $s_{n}=s$
    \item {\em concrete smallest} if it is connected and reduced.
  \end{enumerate}
\end{definition}
\begin{proposition} \label{connected_exists}
  For each coalgebraic Kripke structure $((C, \alpha), I)$, there exists a connected coalgebraic Kripke structure $((C', \alpha'), I')$ such that $((C', \alpha'), I')$ $\doteq$ $((C, \alpha), I)$
\end{proposition}
\begin{proof}
  We set $C'$ be $\{s \in C\ |\ s\ is\ reachable\}$, $\alpha'$ be $\alpha \upharpoonright C'$ and $I'$ be $\{s \in C'\ |\ s \in I\}$. Notably, $((C', \alpha'), I') \doteq ((C, \alpha), I)$
\end{proof}
\begin{definition}[The concrete smallest Kripke structures]
  A Kripke structure is called the concrete samllest Kripke structure if its coalgebraic Kripke structure is the concrete smallest Coalgebraic Kripke structure.
\end{definition}
\begin{lemma} \label{homomorphism_minimal}
Let $((C, \alpha), I)$ be a connected Coalgebraic Kripke structure, $(P, \psi)$ be the final coalgebra of functor $\mathcal{P}_f$ and $h$ be the unique homomorphism from $(C, \alpha)$ to $(P, \psi)$.
\begin{enumerate}
  \item $((C_{\thicksim}, \alpha_{\thicksim}), I_{\thicksim})$ is concrete smallest, where $I_{\thicksim} = \{A | A=\{b \in I| a \thicksim b, a \in I\} \}$
  \item $((C_{K_{h}}, \alpha_{K_{h}}), I_{K_{h}})$ is concrete smallest, where $I_{K_{h}} = \{A | A=\{b \in I| (a, b) \in K_{h}, a \in I\} \}$
  \item $((h(C), \psi \upharpoonright h(C)), h(I))$ is concrete smallest.
  \end{enumerate}
\end{lemma}
\begin{proof}
\begin{enumerate}
  \item $(C_{\thicksim}, \alpha_{\thicksim})$ is extensional.
  \item $K_{h}= \thicksim$  (lemma \ref{bisimulation_equal_kernel})
  \item The proof of $(h(C), \psi \upharpoonleft h(C)) \cong (C_{K_{h}}, \alpha_{K_{h}})$ is trivial, we omit it here.
\end{enumerate}
\end{proof}
\begin{proposition}\label{minimal_iosmorphism}
  If $((M_{1}, \alpha), I_{1})$ and $((M_{2}, \beta), I_{2})$ are two concrete smallest coalgebraic Kripke structures and $((M_{1}, \alpha), I_{1}) \doteq ((M_{2}, \beta), I_{2})$, then $M_{1} \cong M_{2}$
\end{proposition}
\begin{proof}
 It is sufficient to prove that there exists a relation $B \subseteq M_{1} \times M_{2}$ which is a total one-to-one relation. That means for every $s \in M_{1}$, there exists an unique $s' \in M_{2}$ such that $(s, s') \in B$ and for every $s' \in M_{2}$, there exists an unique $s \in M_{1}$ such that $(s, s') \in B$.

 Since $((M_{1}, \alpha), I_{1}) \doteq ((M_{2}, \beta), I_{2})$, there is a bisimulation relation $R$ between $M_1$ and $M_2$. We will prove $R$ is a total one-to-one relation.

 According to definition \ref{coalgebraic_bisimulation_equ}, $I_{1} \subseteq \pi_1(R)$. Because $((M_{1}, \alpha), I_{1})$ is concrete smallest, $((M_{1}, \alpha), I_{1})$ is connected, so $M_{1} \subseteq \pi_1(R)$. Similarly, $M_{2} \subseteq \pi_2(R)$. Simultaneously, $R$ can make the following diagram commute:
\[
\begin{CD}
P @<g<< M_{1} @<{\pi_{1}}<< R @>{\pi_{2}}>> M_{2} @>f>> P\\
@V{\psi}VV @V{\alpha}VV @V{\gamma}VV @VV{\beta}V @VV{\psi}V\\
\mathcal{P}_f(P) @<<\mathcal{P}_f(g)< \mathcal{P}_f(M_{1}) @<<{\mathcal{P}_f(\pi_{1})}< \mathcal{P}_f(R) @>>{\mathcal{P}_f(\pi_{2})}> \mathcal{P}_f(M_{2}) @>>\mathcal{P}_f(f)> \mathcal{P}_f(P)
\end{CD}
\]
where $(P, \psi)$ is final coalgebra of $\mathcal{P}_f(\cdot)$.

$g \circ \pi_{1}$ and $f \circ \pi_{2}$ are both the homomorphism from $(R, \gamma)$ to $(P, \psi)$, so $g \circ \pi_{1} = f \circ \pi_{2}$ (Definition \ref{final_coalgebra})

Suppose there exist $s \in M_{1}$, $s_{1} \in M_{2}$, $s_{2} \in M_{2}$ s.t. $(s, s_{1})$, $(s, s_{2})$ $\in$ $R$ and $s_{1} \neq s_{2}$.
\[
\begin{split}
g(s) &= g \circ \pi_{1}(s, s_{1}) \quad (definition\ of\ \pi_{1})\\
&= f \circ \pi_{2}(s, s_{1}) \quad (g \circ \pi_{1} = f \circ \pi_{2})\\
&= f(s_{1})\quad (definition\ of\ \pi_{2})\\
\end{split}
\]
Similarly, $g(s) = f(s_{2})$, so $f(s_{2}) = f(s_{1})$. According to Theorem \ref{preserves_kernel}, $s_{1} \sim_{M_{2}} s_{2}$, but $s_{1} \neq s_{2}$ as suppose, so $M_{2}$ is not reduced. This contradicts that $M_{2}$ is concrete smallest. So for every $s \in M_{1}$ there is unique $s' \in M_{2}$ s.t. $(s, s') \in R$.
Analogously, for every $s \in M_{2}$ there is unique $s' \in M_{1}$ s.t. $(s', s) \in R$, that is $R$ is a total one-to-one relation. Obviously, there is an isomorphism between $M_{1}, M_{2}$.
\end{proof}
\begin{theorem}
  For every Kripke structure $K$, there exists an unique concrete smallest Kripke structure $M$, such that $M \equiv K$.
\end{theorem}
\begin{proof}
  From Proposition \ref{kripke_structure_coalgebra}, Proposition \ref{connected_exists}, Proposition \ref{minimal_iosmorphism}.
\end{proof}
\section{Examples}\label{Examples}
In this section we will give two examples to find the concrete smallest Kripke structures. The first is finding the concrete smallest Kripke structure for a finite Kripke structure, while the second is for an infinite one. The following method is based on Lemma \ref{homomorphism_minimal}, so in both case the original Kripke structure should be connected.

The basic idea is extracting the image of the original Kripke structure under $h$, where $h$ is the unique homomorphism from the original Kripke structure to the smallest Kripke structure. The image of a state $s \in S$ is an infinite tree obtained by unwinding $K$ from $s$.
\begin{example}
Consider the following Kripke structure.
\begin{center}
  \includegraphics[width=0.3\textwidth]{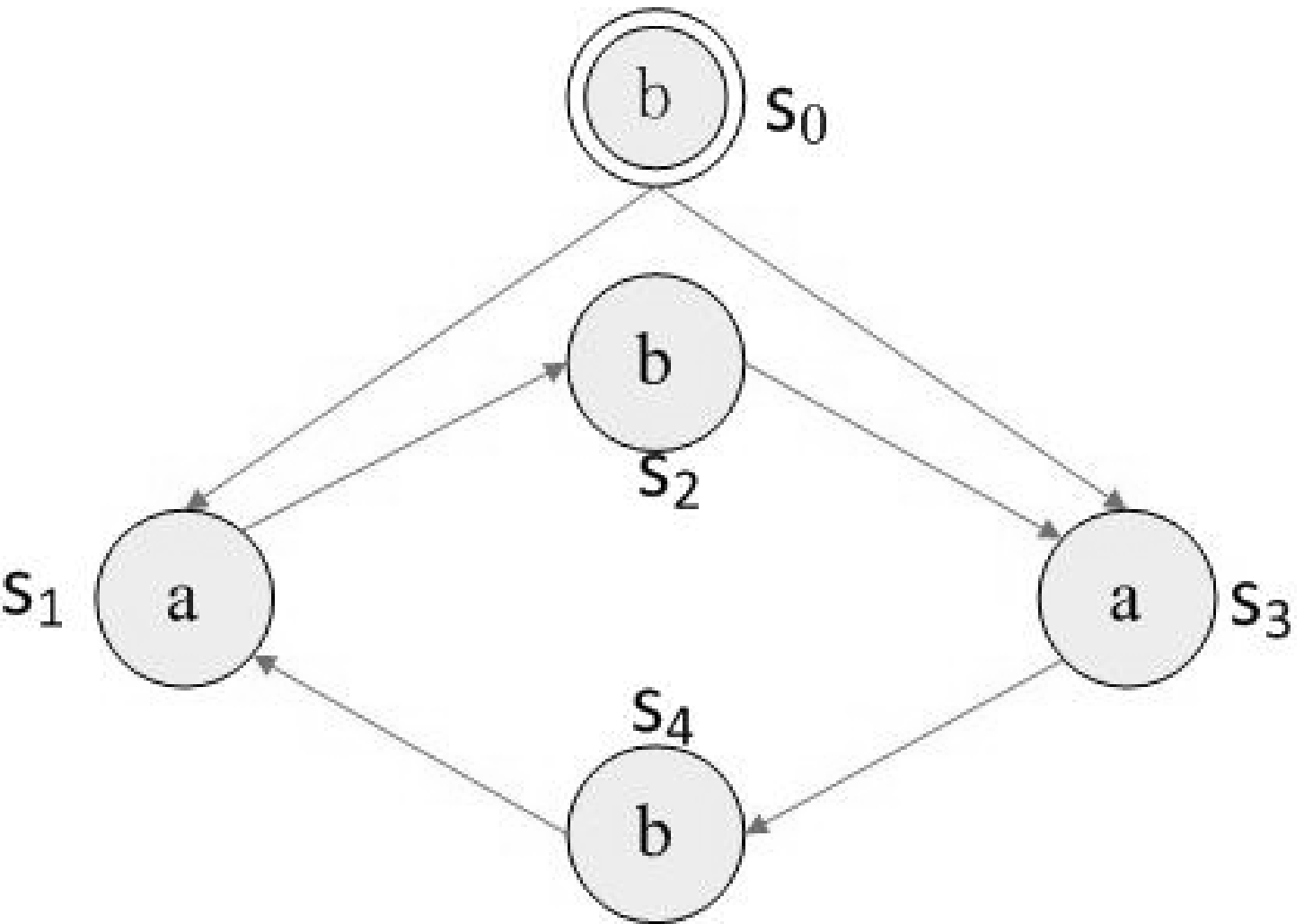}
\end{center}
In first, unwinding this Kripke structure from $s_{0}$ we get the following tree.
\begin{center}
  \includegraphics[width=0.3\textwidth]{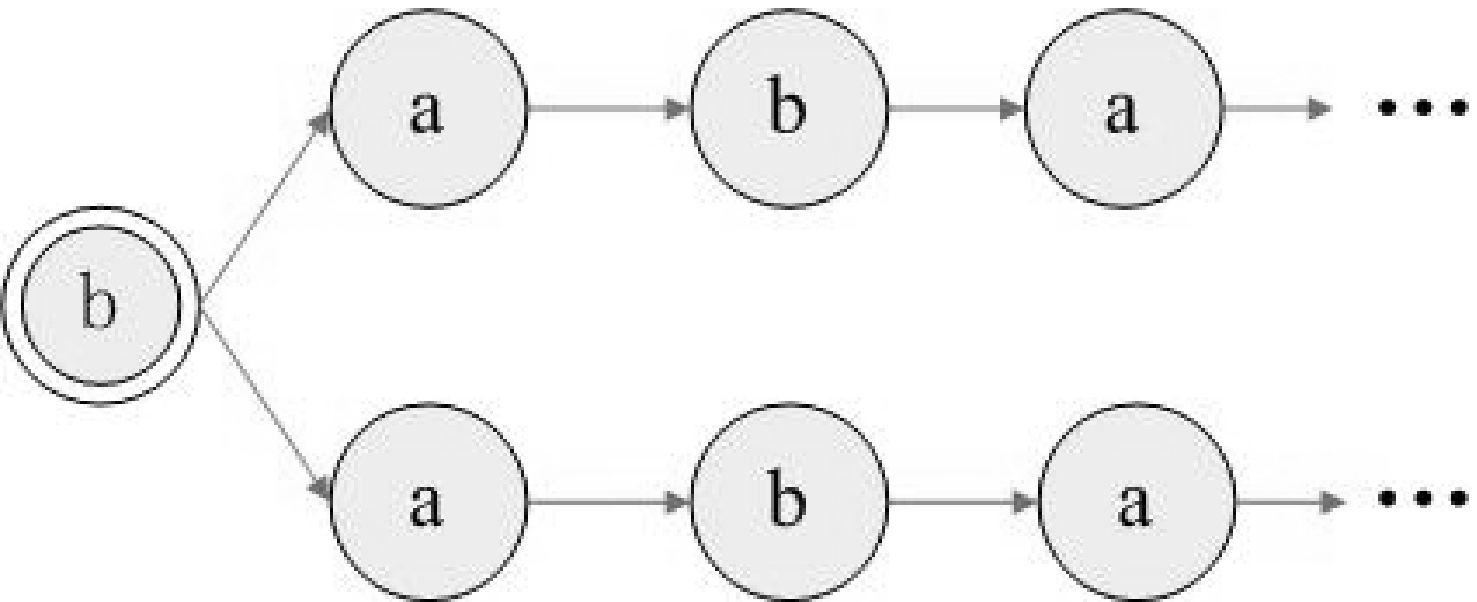}
\end{center}
But its two immediate subtrees are bisimular. So we cutting one off and getting $h(s_0)$:
\begin{center}
  \includegraphics[width=0.3\textwidth]{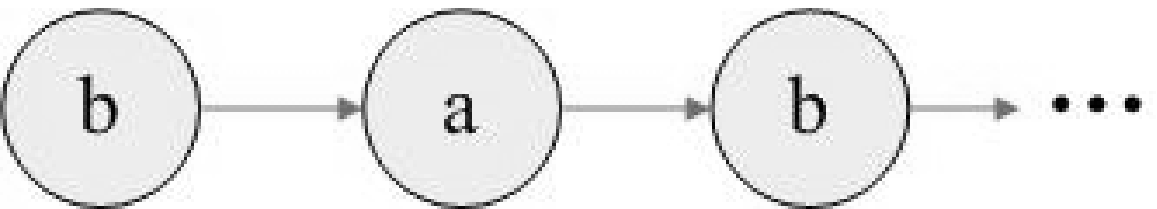}
\end{center}
Then, unwind $s_{1}, s_{3}$. We get the following tree $(h(s_1) = h(s_3))$:
\begin{center}
  \includegraphics[width=0.3\textwidth]{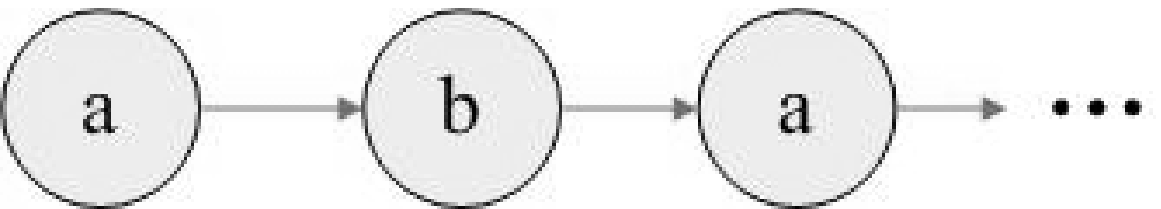}
\end{center}
Finally, unwind $s_{2}, s_{4}$, we get the same tree as unwinding $s_{0}$, that is $h(s_0) = h(s_2) = h(s_4)$.
Now, let us set a new Kripk structure:
\begin{center}
  \includegraphics[width=0.4\textwidth]{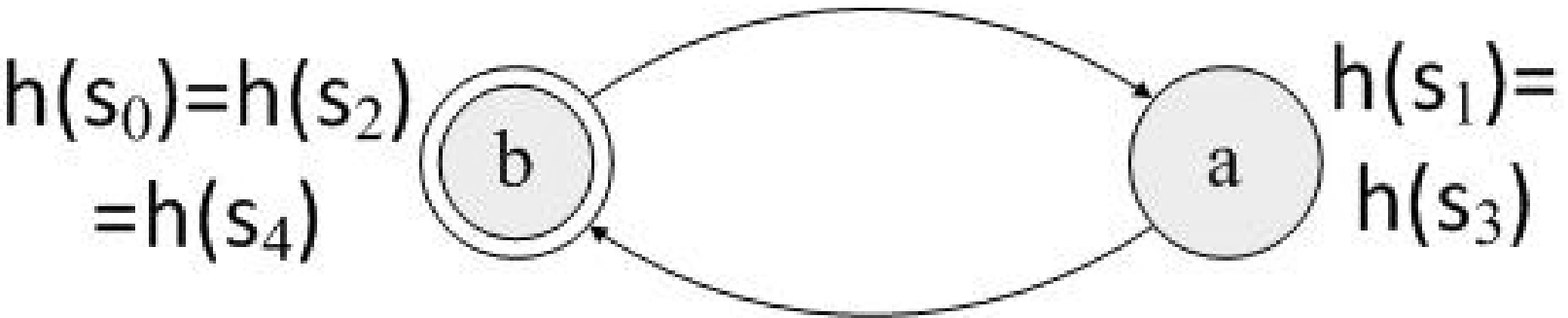}
\end{center}
which is concrete smallest.
\end{example}
\begin{example}
In this case, the Kripke structure is $(\omega, I, R, L)$ over $AP=\{p_0, p_1, p_2,$ $ p_3, p_4\}$, where $\omega$ is the set of natural numbers, $I=\{0\}$, $R = \{(n, s(n))\ |\ s$ is the successor function$\}$, $L(n) = \{ p_i\ |\ i = n\ Mod\ 5 \}$.
In this case, we only need to unwind the Kripke structure for there is only one immediate subtree. So we get the concrete smallest Kripke structure $(S', R', I', L')$, where $S'=\{s_0, s_1, s_2, s_3, s_4\}$, $R'= \{(s_0, s_1),$ $(s_1, s_2), (s_2, s_3), (s_3, s_4), (s_4, s_0)\}$, $I'= \{s_0\}$, $L(s_i)=\{p_i\}$.
\end{example}
This method also can be applied to some Rational Kripke models \cite{Bekker2009Symbolic}. But some Rational Kripke models are infinite branching, so this method does not work well.

\section{Algorithms}\label{Algorithms}
Here are two algorithms finding the concrete smallest Kripke structures for finite Kripke structures and for infinite Kripke Structures defined by Simple Graph Grammars \cite{QuemenerJ95} with some restrictions.

\textbf{The first algorithm} is based on coinduction \cite{Jacobs1997atutorial, Davidebisimulation}.

Given a Kripke structure $M = (S,S_{0},R,L)$
\begin{enumerate}
  \item abandon all not reachable states from S.
  \item $\prod = \{G \subseteq S | s_{1},s_{2} \in G, L(s_{1}) = L(s_{2})\}$;
  \item for $G \in \prod$ do\\
  split $G$ into $G_{1},G_{2},\ldots,G_{n}$. $g_{1},g_{2}$ is in the same subset $G_{i}$ if $\forall G \in \prod R(g_{1})\cap G=\phi$ iff $R(g_{2})\cap G=\phi$. Add $G_{1},G_{2},\ldots,G_{n}$ into $\prod$ and abandon $G$ getting $\prod_{new}$
  \item if $\prod_{new} = \prod$, then let $\prod_{final} = \prod$, otherwise let $\prod = \prod_{new}$ and do 3 again.
  \item We get the concrete smallest Kripke structure $M' = (S',S'_{0},R',L')$
  \begin{enumerate}
    \item $S' = \prod_{final}$
    \item $S'_{0} = \{I \in \prod_{final} | I \cap S_{0} \neq \phi \}$
    \item $R' = \{(G_{1},G_{2}) \in \prod_{final} \times \prod_{final} |\exists g_{1} \in G_{1}, g_{2} \in G_{2}\ and\ (g_{1},g_{2}) \in R\}$
    \item $L'(G_{1}) = L(g_{1})$ where $g_{1} \in G_{1}$
  \end{enumerate}
\end{enumerate}
It is notable that every two elements of  S, $s_{0},s_{1}$ are in one element of $\prod_{final}$ iff $s_{0} \sim s_{1}$, so $M'$ is the concrete smallest Kripke structure of $M$.

\textbf{The second algorithm:} we first transform infinite Kripke Structures defined by Simple Graph Grammars into finite Kripke Structures, then use the first algorithm to act on the finite Kripke Structures.

Let $G_0 =(S_{G_0},R_{G_0},L_{G_0})$ be a finite Kripke structure with N distinguished pairwise distinct states $(ex_i)_1^N$, and $A= (S_A, R_A, L_A)$ a finite Kripke structure with 2N distinguished pairwise distinct states $(in_i)_1^N$ and $(out_i)_1^N$. We further impose that: $\forall i \in [1,N], L_{G_0}(ex_i) = L_A(in_i) =L_A(out_i)$ \cite{QuemenerJ95}.

Let $\overline{G_0} = (G_0, H_{G_0}=\{V ex_1, \ldots, ex_N\})$, and $\overline{A} = (A, H_{A}=\{V out_1, \ldots, out_N\})$ and let $\mathscr{G}_0$ be the structure grammar with the unique rule V in $V in_1, \ldots, in_N \rightarrow \overline{A}$ \cite{QuemenerJ95}.

There should be some restrictions to $(ex_i)_1^N$ and $(out_i)_1^N$. $S_{G_0}/ (ex_i)_1^N$ and $S_A/(out_i)_1^N$ should not be reached by $(ex_i)_1^N$ and $(out_i)_1^N$ respectively. Because the concrete smallest Kripke structure may be infinite if we abandon these restrictions. For example, the concrete smallest of the following Kripke structure is infinite:
\begin{center}
  \includegraphics[width=0.45\textwidth]{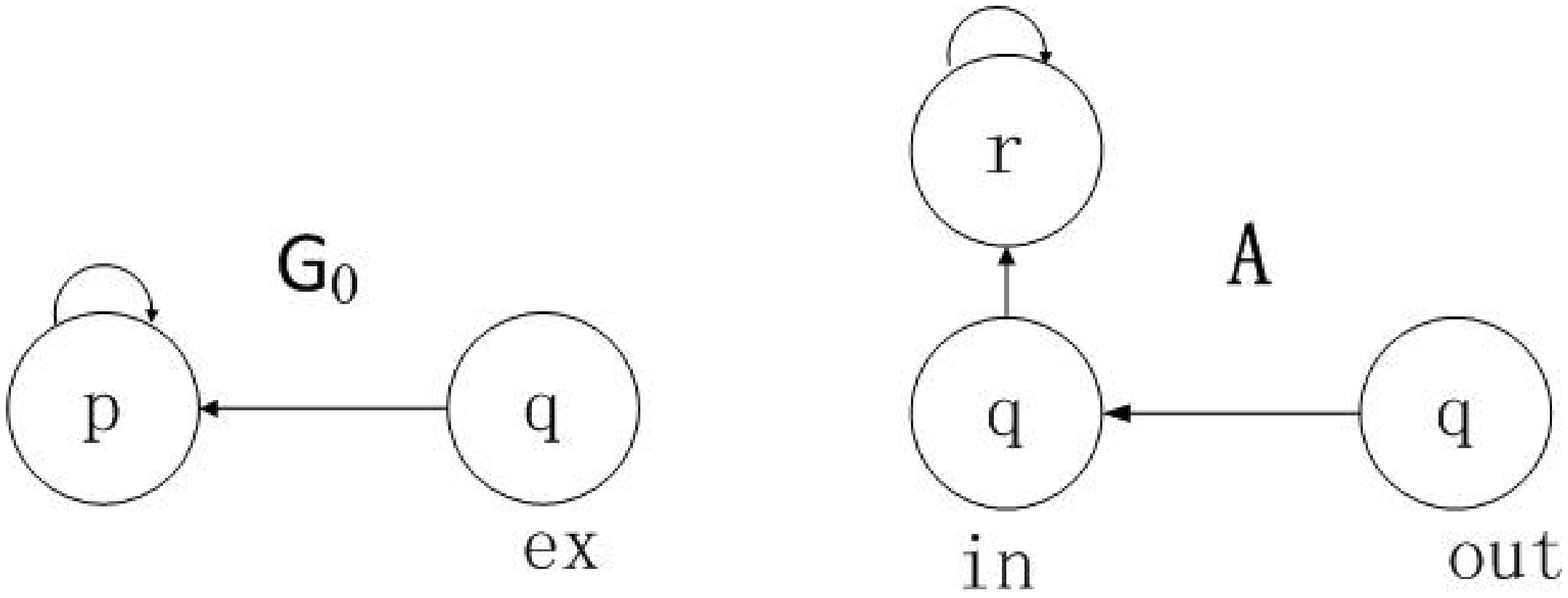}
\end{center}

Now we transform infinite Kripke Structures defined by simple graph grammars with restrictions into a finite Kripke Structure $K = (S,R,L)$
\begin{enumerate}
    \item $S = S_{G_0}  \cup ( S_A/(in_i)_1^N)/(out_i)_1^N$
    \item $R = R_{G_0} \cup \{(s_1,s_2) \in R_A | s_1,s_2 \notin (in_i)_1^N \cup (out_i)_1^N\}\cup\{(s_1, ex_i) | (s_1, in_i) \in R_A\ or\ (s_1, out _i) \in R_A, \ i \in[1,N]\} \cup \{(ex_i,s_1)| (in_i,s_1) \in R_A\}$
    \item $L(s) = L_{G_0}(s)$, if $s \in G_0$; $L(s) = L_A(s)$, if $s \in S_A$.
\end{enumerate}
The next step is using first algorithm to act on the Kripke structure $K$, and we will get the concrete smallest Kripke structure of the infinite Kripke structure with $(G_0,A)$ as its finite representation.

\end{document}